\documentclass{llncs}

\RequirePackage[margin=1.2in]{geometry}
\usepackage{mathtools}
\usepackage{amsmath}
\usepackage{amssymb}
\usepackage{epsfig}
\usepackage{float}
\usepackage[justification=centering]{caption}
\usepackage{color}
\usepackage{textpos}
\usepackage{multicol}
\usepackage{algorithm}
\usepackage{algpseudocode}
\usepackage{caption}

\def\N{{\mathbb{N}}}

\begin{document}

\title{Balancing Communication for Multi-party Interactive Coding}
\author{Allison Lewko\inst{1} \and Ellen Vitercik\inst{2}}

\institute{Columbia University, New York NY 10027, USA,\\
\email{alewko@cs.columbia.edu}
\and
Columbia University, New York NY 10027, USA,\\
\email{emv2126@columbia.edu}}

\maketitle     

\begin{abstract}
We consider interactive coding in a setting where $n$ parties wish to compute a joint function of their inputs via an interactive protocol over imperfect channels. We assume that adversarial errors can comprise a $\mathcal{O}(\frac{1}{n})$ fraction of the total communication, occurring anywhere on the communication network. Our goal is to maintain a constant multiplicative overhead in the total communication required, as compared to the error-free setting, and also to balance the workload over the different parties. We build upon the prior protocol of Jain, Kalai, and Lewko, but while that protocol relies on a single coordinator to shoulder a heavy burden throughout the protocol, we design a mechanism to pass the coordination duties from party to party, resulting in a more even distribution of communication over the course of the computation.
\end{abstract}

\section{Introduction}

The fundamental problem of errors in communication has been studied ever since the groundbreaking work of Shannon \cite{Sha48}. Due to his results, we know how to construct error-correcting codes that achieve a constant information rate despite a constant error rate. In this paper, we study error correction in the interactive setting, an area that was introduced by Schulman \cite{Sch92}, \cite{Sch93}. In particular, we are interested in this problem in the multi-party setting.

We note that it is insufficient to simply encode each message of the original interactive protocol with a traditional, non-interactive error-correcting code. In the case of adversarial noise, this will result in a poor error rate, since  if $k$ messages are sent in the protocol, the error-rate must be less than $1/k$. After all, we can not allow even a single message to be fully corrupted without compromising the correctness of the computation.

In this work, we show how to convert any $n$-party interactive protocol into a new protocol that is resilient to $\Theta(\frac{1}{n})$-fraction of adversarial error, while incurring only a constant blow-up in the communication complexity $CC$.

This is similar to the properties achieved by the protocol in \cite{JKL}, but we achieve a new additional feature. Namely, \cite{JKL} requires a single party to act as a ``coordinator" of the entire computation, and that party must individually perform a constant fraction of the overall total communication. In a distributed setting with $n$ parties, we might naturally want the communication burden to be more equally distributed over the course of the protocol. To ease the burden of central coordination, we introduce new techniques that allow the role of the coordinator to be passed from party to party in a rotating fashion, equalizing the communication over time. Thus, we do not require that any one party $P_i$ ``lead'' the computation for the entire time, and this avoids the resulting blowup in $P_i$'s communication complexity.

\subsection{Prior Work in the Multiparty Setting}

Rajagopalan and Schulman \cite{RS94} first extended the problem of error-resilient interactive coding to the multi-party setting, and showed how to achieve error-resilience against stochastic errors. When there are $n$ parties communicating, their method has a communication complexity overhead of $\Theta(n^2 \log n)$. The original protocol is converted to one that proceeds in rounds such that every party $P_i$ sends a message to all of its neighbors during each round. Therefore, if the original protocol is synchronous, i.e. $\Omega(n^2)$ bits are transmitted per round, and there is a $\Theta(\log n)$-factor overhead in the number of rounds.

However, protocols are not always synchronous, an issue which Jain, Kalai, and Lewko addressed \cite{JKL}. They consider an arbitrary $n$-party protocol $\pi$ with static speaking order and at least one party $P$* who shares a communication link with all other parties. In this context, they present a compiler that converts $\pi$ into a new protocol $\tilde{\pi}$ that is resilient to a $\Theta(\frac{1}{n})$-fraction of adversarial errors and which incurs only a constant blow-up in the communication complexity. Jain et al. require that every party sends all of its outgoing messages in $\pi$ to $P$*, who delivers each message to its recipient. $\pi$ is thus restricted to pairwise interactions with $P$*, all of which are protected via a two-party interactive coding scheme. Whenever $P$* notices an error on a channel, he signals a rewind to every party. It is important to note, however, that if a party $P_i$ is silent for too long a period, $P$* may not realize that there was an error on his channel with $P_i$, in which case the computation may continue, incorrectly, for an extended period of time. Therefore, Jain et al. intersperse phases during which the parties exchange regular protocol messages with ``polling phases,'' when $P$* checks in with each party about its current state. Though $CC(\tilde{\pi}) = c \cdot CC(\pi)$ for some constant $c$, the increase in work required of $P$* is potentially undesireable.

\subsection{Our Results}

We present a compiler that converts any $n$-party protocol a new protocol that achieves a constant information rate and is resilient to a $\Theta(\frac{1}{n})$-fraction of adversarial errors. Moreover, our scheme does not incur a blowup in the communication complexity of any one party during the computation. We assume that the network is complete and that the speaking order in the noiseless protocol is static.

\begin{theorem}\label{main_thm} (Informal)
There exists a compiler {\normalfont\textsf{Comp}} and constants $c > 1$ and $\epsilon \in (0,1)$ such that for any $n$-party protocol $\pi$ with static speaking order and complete communication links, {\normalfont\textsf{Comp}} compiles $\pi$ into a new protocol $\tilde{\pi}$ such that:

\begin{enumerate}
\item $CC(\tilde{\pi}) = c \cdot CC(\pi)$.
\item Protocol $\tilde{\pi}$ is resilient to $\frac{\epsilon}{n}$-fraction of adversarial errors (with high probability).
\item The communication complexity of any party $P_i$ is $\Theta\left(\frac{CC(\pi)}{n}\right)$.
\end{enumerate}
\end{theorem}

Item 3. above only holds when the communication of the original protocol is balanced among the parties, and the base protocol of \cite{JKL} can be suitably modified to preserve this, except for the undue burden it places on the coordinator. We discuss this in more detail in Section \ref{sec:JKL}.

\paragraph{Our Techniques}
In \cite{JKL}, the special party that coordinates the computation simulates the error-free computation and maintains a global view, rewinding when necessary to correct errors and frequently speaking with all other parties to detect any errors that have not yet been revealed. A natural idea for distributing the work of this special party is to break the computation into ``chunks" (similar to \cite{BK12}) and have a different party serve as the coordinator for each individual chunk simulation. The basic difficulty in implementing this approach is that we need a new mechanism for efficiently checking for past errors. In \cite{BK12}, where chunks are also simulated individually, hashes of the entire simulated transcript so far are used to determine if the previous chunks have been properly simulated. If these hashes reveal inconsistencies in the old chunks, the simulation process rewinds and re-simulates older chunks under the discrepancies are resolved. This is crucial to success against a significant error rate, as many individual chunk simulations may be corrupted, and moving on from these and not periodically rechecking them will lead to an incorrect result.

However, we cannot simply use hashes of the entire past transcript to check consistency as in \cite{BK12}, because our chunk coordinator is now rotating, and does not have a full view of the past chunk simulations when it was not serving as the chunk coordinator. To address this issue without incurring a super-constant blowup in total communication, we apply hashing not to the naked transcripts, but to previous hashes concatenated with the most recent chunk transcript. These previous hashes can be passed from an old coordinator to a new one at each chunk without increasing the communication complexity by more than a constant factor. One subtlety in ensuring that the communication load per party remains balanced over time is that the same party may be asked for the same old hashes more than once as the protocol attempts to resimulate a particular chunk multiple times. We address this by allowing a party to refuse to communicate under certain circumstances and for the protocol to continue anyway via a timeout mechanism. Our analysis further shows that these refusals do not prevent our simulation from making progress when the error rate is suitably bounded.

Ultimately, the compression of previous transcripts into iterated hashes plus the message passing of hashes between coordinators results in a more evenly distributed protocol without sacrificing total communication complexity or error rate (up to constant factors). We view this as a necessary first step in adapting and expanding multi-party interactive coding techniques to be more appropriate for a truly distributed setting. Ultimately, we would like to see such techniques extended to achieve stronger error-resilience for a wider variety of multi-party tasks, including classical distributed computing tasks like Byzantine agreement.

\subsection{Additional Related Work}

There have been several works improving on the (1/240)-fraction of adversarial of errors allowed by Shulman's original (two party) protocol. Braverman and Rao  improved significantly on the tolerable error rate, allowing $\frac{1}{8} - \epsilon$ with a binary alphabet or $\frac{1}{4} - \epsilon$ with a constant alphabet \cite{BR11}. \cite{AGS13} and \cite{GHS14} have improved the error rates beyond $\frac{1}{4}$ by leveraging adaptivity.

Both of the compilers in \cite{Sch93} and \cite{BR11} rely on tree-codes, which we do not know how to construct or decode efficiently. Recent work has made progress toward the efficient construction of tree-codes. In the stochastic case, Gelles, Moitra, and Sahai \cite{GMS11} showed that a weaker form of tree codes is sufficient, and thus improved on the protocols of \cite{Sch92}, \cite{Sch93}, and \cite{RS94}.

Brakerski and Kalai \cite{BK12} considered the problem with adversarial error and presented an efficient version of Schulman's compiler. \cite{BN13} improved upon the computation complexity of Brakerski and Kalai's construction. More recently, \cite{GH14} provided a simple scheme for efficiently simulating any two-party protocol, achieving optimal communication and error rates. The works of \cite{GH14,BE14} also consider list-decoding for interactive communication, while the works of \cite{KR13,H14} study the channel capacity for interactive communication.

\section{Preliminaries}

In this section, we fix notations and definitions that we will rely on throughout the paper. We will begin with description of the noiseless protocol and then move on to a description of the two-party protocol from \cite{Sch93} and the multiparty protocol from \cite{JKL}, both of which we will rely on in our own construction.

\subsection{The Noiseless Protocol} Let $\pi$ be the noiseless protocol for the $n$ parties $P_1, \dots, P_n$. We consider the case where each message in $\pi$ consists of a single bit, since this is the ``hardest'' case. Moreover, we assume that messages are sent sequentially, so that only when a party $P_i$ receives a message from some $P_j$ does $P_i$ send a new message to some $P_k$, as dictated by the static speaking order of $\pi$. We assume that the first party to speak in $\pi$ does so after receiving a dummy message. At any point during the computation of $\pi$ when it is some $P_i$'s turn to speak we can describe the \emph{next-message function} \textsf{NM$_i$}. Let $X = x_1, x_2, \dots, x_n$ denote the inputs of $P_1, P_2, \dots, P_n$ and $B_{\ell-1} = b_1, \dots, b_{\ell - 1}$ denote the first $\ell - 1$ messages of $\pi$. Finally, let $\textsf{trans}_i$ be the partial protocol transcript observed by $P_i$. Then $b_\ell = $ \textsf{NM$_i$}$(x_i; \textsf{trans}_i)$.

Let $L = CC(\pi)$. The entire transcript of $\pi$ corresponds to a unique path from the root to a leaf of a binary tree of depth $L$, which we will denote by $\mathcal{T}$. Each node of $\mathcal{T}$ corresponds to a party $P_i$, and each arc corresponds to a message sent in $\pi$, either a 0 or a 1. If there is an arc from a node labeled $P_j$ of depth $\ell$ to a node labeled $P_k$ via an arc labeled $b$ on the path in $\mathcal{T}$ corresponding to the transcript of $\pi$, then the $\ell$th message of $\pi$ is the bit $b$ sent from $P_j$ to $P_k$.

\subsection{More Details of Our Model}

We consider an adversary that's computationally unbounded and can flip bits anywhere. The adversary is only constrained to having a specific error budget over the lifetime of the protocol.

In our measurement of communication complexity, we allow for a timeout mechanism, whereby a party that does not want to send requested information can signal that by not replying  in a fixed amount of time. We do not count this in the communication complexity.

Finally, William Hoza \cite{Hoza14} observed that adding links between parties can be used to tolerate an arbitrarily high error-rate. If Alice and Bob are connected by two channels and the adversary cannot insert or delete bits, but only alter their contents, then Alice can communicate `0' to Bob by only using the first channel, and `1' by only using the second. In this way, Bob can decode Alice's messages perfectly even without looking at the contents. We assume that there is only one path between parties, so we leave this as a topic for future work.

\subsection{Schulman's Two-Party Compiler} Here we give a brief overview of a slight variant on Schulman's  compiler \cite{Sch93}, as described in \cite{JKL}. Let $\pi = \langle P_1, P_2 \rangle$ be any two-party protocol. For the error-resilient protocol $\tilde{\pi}$,

Let $\pi$ be the noiseless protocol and $\mathcal{T}$ be the protocol tree for $\pi$. For the error-resilient protocol $\widetilde{\pi}$, each party is equipped with a pebble $\alpha_i$ which points to a node in $\mathcal{T}$. $\mathcal{T}$ is padded to include dummy nodes at the bottom of the tree so that the height of the tree equals the length of the simulation. At any point, the movement of $\alpha_i$ can be described by 0, meaning that it moved down to the left child, 1, meaning that it moved down to the right child, $H$, meaning that it stayed put, and $B$, meaning that it moved back to the parent node. The movement of each pebble can be described by a 4-ary tree where each arc is labeled 0, 1, $H$ or $B$. This 4-ary tree is called a \emph{history tree}, and is denoted $\mathcal{HT}$. The two parties also share a 4-ary tree code $\mathcal{TC}$ of depth $N$ over an alphabet $\Sigma$ of constant size $c$, which they use to encode and decode each others' pebble movements. We are now ready to describe the steps taken by $P_1$ upon receiving a message from $P_2$, which are symmetric to those taken by $P_2$ upon receipt of a message from $P_1$.

\begin{enumerate}
\item \textbf{Guess $P_2$'s pebble position:} Let $\vec{\nu} = \nu_1, \dots, \nu_k$ be the sequence of tree-code symbols $\nu_i$ that $P_1$ has received from $P_2$ so far. $P_1$ guesses the history $\widetilde{hist_2}$ of pebble moves made by $P_2$ such that the hamming distance $\Delta(\mathcal{TC}(\widetilde{hist_2}, \vec{\tau})$ is minimized. $P_1$ uses $\widetilde{hist_2}$ to compute $\widetilde{\alpha_2}$, its guess for $P_2$'s pebble position in $\mathcal{T}$.

\item \textbf{Compute the next pebble move:} There are several cases depending on the relative positions of $\alpha_1$ and $\widetilde{\alpha_2}$.

\begin{itemize}
\item If $\alpha_1 = \widetilde{\alpha_2}$, then in $P_1$'s view, $P_1$ and $P_2$ are in the same position in $\mathcal{T}$. If it is $P_1$'s turn to speak, then he computes $\tau = \textsf{NM}_1(x_1, \textsf{trans}),$ where $\textsf{trans}$ is the concatenation of the arcs along the path from the root of $\mathcal{T}$ to $\alpha_1$. Otherwise, it is $P_2$'s turn to speak, in which case $P_1$ sets $\tau = H$.
\item Otherwise, if $\alpha_1$ is the parent of $\widetilde{\alpha_2}$ and $P_2$ is the label of the node pointed to by $\alpha_1$, then $P_2$ has sent a message to $P_1$, and $P_1$ should set $\tau = 0$ (respectively, 1) if $\widetilde{\alpha_2}$ is the left (respectively, right) child of $\alpha_1$.
\item Otherwise, if $\alpha_1$ is an ancestor of $\widetilde{\alpha_2}$, then $P_2$ may have moved along an incorrect path in $\mathcal{T}$. Therefore, $P_1$ sets $\tau = H$ and waits for $P_2$ to move back up the tree.

\item Otherwise, if the least common ancestor of $\alpha_1$ and $\widetilde{\alpha_2}$ is a strict ancestor of $\alpha_1$, then in $P_1$'s view, $P_1$ and $P_2$ have diverged onto different paths in $\mathcal{T}$. Therefore, $P_1$ sets $\tau = B$ in order to back up to the point of consistency in $\mathcal{T}$.
\end{itemize}

Now that $P_1$ has computed its next pebble move $\tau$, it moves $\alpha_1$ accordingly.

\item \textbf{Send next symbol to $P_2$:} Let $\textsf{hist}_1$ be the history of $P_1$'s pebble moves made during the computation of $\tilde{\pi}$. $P_1$ computes $\vec{\nu} = \mathcal{TC}(\textsf{hist}_1)$ where $\vec{\nu} = \nu_1, \dots, \nu_{|\textsf{hist}_1|}$. $P_1$ sends $\nu_{|\textsf{hist}_1|}$ to $P_2$.
\end{enumerate}

\subsection{Jain et al.'s Multi-Party Compiler}\label{sec:JKL}
First, we give the main theorem from \cite{JKL} and then give a high-level overview of their compiler.

\begin{theorem}[\cite{JKL}]\label{JKL_thm}
There exists a compiler {\normalfont\textsf{Comp}} and constants $c > 1$ and $\epsilon' \in (0,1)$ such that for any $n$-party protocol $\pi' = \langle P_1, \dots, P_n \rangle,$ with static speaking order and at least one party $P_i$ that shares a communication link with all parties $\{P_j\}_{j \not= i}.$ {\normalfont\textsf{Comp}} compiles $\pi'$ into a new protocol $\tilde{\pi}'$ such that:

\begin{enumerate}
\item $CC(\tilde{\pi}') = c \cdot CC(\pi').$
\item $\tilde{\pi}'$ is resilient to $(\frac{\epsilon'}{n})-$fraction of (adversarial) errors in the total communication.
\item The runtime of each party $P_i$ is at most $2^{\mathcal{O}(n \cdot C_{max})},$ where $C_{max}$ is the maximum number of bits sent and received by any party $P_j$ in the underlying protocol $\pi'$.
\end{enumerate}
\end{theorem}

Jain et al. designate one party $P$* who shares a communication link with all other parties to ``lead'' the protocol. The network operates as a star network with $P$* as the hub node. $\pi'$ is thus decomposed into pairwise protocols $\pi_1, \dots, \pi_n$ where $\pi_i$ denotes the two-party protocol for communication between $P$* and $P_i$. $P$* shares the protocol tree $\mathcal{T}_i$ for $\pi_i$ with each $P_i$. $P$*'s copy of this tree is denoted $\mathcal{T}_i$*.

\cite{JKL} splits the simulation of $\pi$ into a sequence of phases. The phases alternate between \emph{protocol phases} and \emph{polling phases.} $\Theta(n)$ bits are exchanged during each phase. During the protocol phases, each $P_i$ follows the two-party protocol simulation strategy, since he is only communicating through $P$*. Meanwhile, every time $P$* receives a message, he computes the ``global consistency point'' in the protocol tree $\mathcal{T}$ for $\pi$. If this is further up the tree than the current node of the simulation, then $P$* signals a rewind, telling parties to back up until every party has returned to the section of $\mathcal{T}$ that is error free. During the polling phases, each party $P_i$ sends a tree code symbol indicating its pebble position in the two-party protocol tree for $\pi_i$. Therefore, if some party speaks much less often compared to other parties, an error in his computation history will still be detected in a timely manner.

In order to achieve an overall balance of communication among the players, we will need to use a variant of \cite{JKL}'s compiler that achieves such a balance, up to the imbalanced burden played on the special coordinator $P^*$. Naturally, one cannot hope to achieve a constant multiplicative overhead in the total communication complexity and a $\mathcal{O}(\frac{1}{n})$ relative complexity for each party if the original (noiseless) protocol does not have balanced communication. However, even when the original protocol does have relatively balanced communication complexity over the $n$ parties, the protocol of \cite{JKL} does not preserve this. The glaringly obvious violation is the disproportionate burden on $P^*$, but there is a more subtle potential violation as well. As $P^*$ attempts to simulate the error-free transcript during the non-polling phases of the protocol, it may adaptively speak with whatever party it deems relevant based on its current (possibly wrong) view, potentially causing communication to become unbalanced among the regular parties.

There are several cases in which this undesirable behavior of the \cite{JKL} protocol can be avoided. One case is if the underlying (noiseless) protocol is actually synchronous, proceeding in rounds in which each party sends the same number of bits. In such a setting, $P^*$ can simply collect the messages from all parties for each round and simulate rounds in a balanced way. Polling phases can then be eliminated and the analysis given in \cite{JKL} to ensure successful simulation with an error rate of $\Theta(\frac{1}{n})$ and a constant multiplicative overhead in total communication complexity still applies.

Another case is where the underlying protocol is asynchronous, but proceeds in balanced windows of communication where each party speaks once (and sends the same number of bits). For example, consider a protocol that consists of $P_1$ sending a bit, then $P_2$, then $P_3$, and so on, simply going through the parties one by one in a fixed order. The coordinator $P^*$ of the \cite{JKL} protocol could then proceed to speak with parties in this same ordering, and again separate polling phases could be eliminated.

More precisely, the $P^*$ in this case will still maintain a view of a ``global consistency point" in the overall protocol tree where he believes the other parties need to rewind too, but he will work in units that correspond to the communication windows and stick to the fixed speaking order while attempting to rewind parties back to this point and learn new information (so the global consistency point will always be defined as the beginning of a window). Intuitively, since the \cite{JKL} protocol can tolerate waiting for a polling phase for $P^*$ to speak to a particular party and discover a previous error, this approach can also tolerate a $\Theta(\frac{1}{n})$ error rate. 

To see this more formally, we can tweak the analysis given in \cite{JKL} by defining an adjusted measure of progress. Our measure $M$ could be defined (in window units) as follows. We let $depth(gcp)$ denote the window index of the true global consistency point in the noiseless $n$-party protocol tree. For each party $P_i$, we let $d(i,gcp)$ denote the number of number of windows away from $gcp$ that $P^*$'s and $P_i$'s pebbles are in their 2-party simulation. 
We then define:
\[M := depth(gcp) - \max_{i \in [n]} \{ d(i,gcp)\}\]

Analogously to \cite{JKL}, we can then define a good window as a window where all the symbols are received correctly by $P^*$ and each party $P_i$ and that all parties have correctly guessed the positions of all pebbles. In such a window, we claim that $M$ will strictly increase. To see this, we note that if $d(i,gcp) = 0$ for all $i$, then $P^*$ will successfully simulate a new window, and $depth(gcp)$ will increase. If $d(i,gcp) >0$ for some values of $i$, all of these parties will move their pebbles appropriately during the window in order to decrease $\max_{i \in [n]} \{ d(i,gcp)\}$. 

We similarly claim that any non-good window can only decrease $M$ by a constant amount. First, it is clear that $depth(gcp)$ can only change by a constant amount, and since $d(i,gcp)$ only changes by a constant amount for each $i$, the max can also only change by a constant. 

Now an analysis common to \cite{JKL} and \cite{Sch92} easily applies: we can take every non-good window and argue that it is contained in a \emph{bad interval} that contains a $\Omega( \frac{1}{n})$ fraction of errors. If any symbol in the window itself is corrupted, then that window serves as a suitable bad interval (since $\mathcal{O}(n)$ symbols are sent per window). Otherwise, the constant rate of the tree code and the fact that every party sends one symbol per window means that we take a backward-stretching interval including a number of windows proportional to the depth of the worst tree-decoding error as a suitable bad interval. As in \cite{JKL}, we can then see that a constant multiplicative overhead in the number of window simulations then suffices to insure a correct simulation of the underlying protocol. 

It is worth noting that our approach for passing off the duties of $P^*$ from party to party is rather modular, and does not depend upon the precise details of how $P^*$ simulates a piece of the protocol. Thus, other variants of \cite{JKL} or other base protocols that equalize communication complexity among parties \emph{except for a heavy burden on $P^*$} could also be inserted into our protocol to obtain analogous results.

\subsection{Hash Functions} We use a family of hash functions index by keys $k \in \{0,1\}^ t$. In particular, we invoke the following theorem of \cite{NN93}, \cite{AGHP92}.

\begin{theorem}\label{hash_thm}[\cite{NN93}, \cite{AGHP92}]
There exists a constant $q > 0$ and an ensemble of hash families and an ensemble of hash families $\{H_N\}_{N \in \N}$ such that for every $N \in \N$ and for every $h \in H_N,$ $h:\{0,1\}^{\leq 2^N} \to \{0,1\}^{qN}$ is poly-time computable, it is efficient to sample $h \leftarrow H_N$ using only $qN$ random bits, and for all $y \not=z \in \{0,1\}^{\leq 2^N}$ it holds that \[\underset{h \leftarrow H_N}{\textnormal{Pr}}[h(y) = h(z)] \leq 2^{-N}.\]
\end{theorem}

We set $t = qN$ and write $h_k:\{0,1\}^{\leq 2^N} \to \{0,1\}^t$ to denote the element of $H_N$ sampled with the random string $k \in \{0,1\}^t$. We also let $\{Enc_1, Dec_1\}$ and $\{Enc_2, Dec_2\}$ denote two pairs of encoding and decoding algorithms of error-correcting codes with constant rates $\beta$ and a constant relative distance $\lambda$. In particular, we have \[Enc_1 : \{0,1\}^{2t} \to \{0,1\}^{2\beta t}\] and \[Enc_2 : \{0,1\}^{nt} \to \{0,1\}^{\beta nt}.\]

\section{Our Compiler}

\subsection{Overview}

Our compiler consists of two main phases: \emph{chunk simulations} and \emph{consistency checks.} These two phases make up a single ``iteration'' of $\tilde{\pi}$. A \emph{chunk} refers to a section of $\mathcal{T}$ such that the chunk indexed by $j$ includes all nodes in $\mathcal{T}$ of depth greater than or equal to $jk$ and less than $(j+1)k$ for a parameter $k$ that we will define. The parties will sequentially take turns serving as $P$* for a chunk simulation, during which time the protocol from \cite{JKL} (adapted as described in Section \ref{sec:JKL}) will be run. Each $P_i$ is equipped with variable $\gamma_i$, set to 0 at the beginning of the protocol execution, which indicates the chunk that the party is currently simulating. Also, as in \cite{JKL}, each party is equipped with a pebble $\alpha_i$ which points to the root of $\mathcal{T}_i$, as well as pebbles $\alpha_i$* with which to lead the simulation during their turn as $P$*.

During the consistency checks, the parties ensure that they are all simulating the same chunk of $\mathcal{T}$ and they rewind if they are out of sync. If there is a rewind, the new leader $P_{new}$* must be able to request information from the old leader $P_{old}$* who led the simulation during the chunk $P_{new}$* has backed up to. We later refer to this as ``tapping'' $P_{old}$*. $P_{new}$* must have a way of checking that the \emph{entire} computation of $\pi$ has been correct, even though he was not the leader for most of the chunks. Therefore, we require each party to store information about the computation of each chunk so that they can ensure the correctness of the simulation so far. This data, which we will describe in detail later in this section, is stored in a vector, which we call $\vec{\rho}_i$ for each $P_i$. $P_i$'s data regarding the $j$th chunk of $\mathcal{T}$ is stored in $\vec{\rho}_i$[$j$]. $P_i$ may write over data stored in any $\vec{\rho}_i$[$j$] if the $j$th chunk of $\mathcal{T}$ is simulated multiple times during the computation.

Let $CC(\pi) = L$. The simulation is stopped after a total of $5L$ symbols have been exchanged. We let $m$ be a constant positive integer such that it takes $mk$ exchanges to simulate a chunk using the protocol from \cite{JKL}. There are at least $4n \beta t$ bits exchanged during the remainder of an iteration of $\widetilde{\pi}$, where $\beta$ is the constant rate of the encoding algorithms and $t$ is the length of the hash function output. We set $4n\beta t = mk$, so that the number of bits exchanged during the simulation of a chunk via the protocol in \cite{JKL} is at least the number of bits exchanged during the remainder of an iteration in $\widetilde{\pi}$. Therefore, $k$, the depth of each chunk, is set to be $4n \beta t/m$.

At the beginning of a chunk, $P$* will begin communication with the $P_k$ that is the label of the first node of the chunk in $\mathcal{T}$. For the remainder of the chunk simulation, communication among parties will follow the protocol described in \cite{JKL}.

We will first describe the initial iteration of $\tilde{\pi}$ and then an arbitrary $j$th iteration of $\tilde{\pi}$ for $j > 1$.

\subsection{The first iteration} \label{firstTransfer}

$P_1$, as $P$*, will progress through the first chunk of $\mathcal{T}$, following the protocol described in \cite{JKL}. Now we give a high level overview of the first consistency check. Each $P_i$ hashes his transcript from $\mathcal{T_i}$ and stores it in $\vec{\rho}_i$[$0$]. In the next iteration, he will concatenate this hashed value with the transcript from the second chunk and store it in $\vec{\rho}_i$[$1$], and so on. During the consistency check, each party $P_i$ sends this hashed value to $P$* and $P$* computes his own version of this hash value for each $P_i$. If the hashed values all match up with $P$*'s $n$ computed values, then $P$* sends $F$ (``forward'') to each party. Otherwise, $P$* tells the parties to back up to the beginning of the chunk by sending $B1$ (``back up one chunk'') to each party. (In future iterations, $P$* may send $B2$ to a party $P_i$ in order to have the party back up to chunk number $\ell - 1$ if $\gamma_i = \ell$.) Specifically, the following steps occur during the first consistency check:

\begin{enumerate}
\item \textbf{Each $P_i$ computes its hash value and sends it and $\gamma_i$ to $P$*:} For all $i \in [2, \cdots, n]$, let $\sigma_{i,0}$ be the local transcript of the two-party protocol in $T_i$ corresponding to the first chunk of $\mathcal{T}$. Set $\psi_{i,0} \vcentcolon= \sigma_{i,0}.$ $P_i$ samples $k_i \leftarrow H$ and sends $Enc_1(H_{k_i}(\psi_{i,0}) \ || \ k_i)$ to $P$*. We denote received hash values by adding a tilde above the $H$, so we say that $P$* receives $\widetilde{H}_{k_i}(\psi_{i,0}$) from each $P_i$.

\item \textbf{$P$* computes its corresponding hash value:} For all $i \in [2, \cdots, n]$, let $\psi_{i,0}$* be $P$*'s local transcript of the two-party protocol in the first chunk of $\mathcal{T}_i$*. $P$* computes the corresponding $H_{k_i}(\psi_{i,0}$*).

\item \textbf{$P$* compares the hash values:} We say that the chunk was \emph{good} if for all $i$, $\widetilde{H}_{k_i}(\psi_{i,0}) = H_{k_i}(\psi_{i,0}$*). Otherwise, we say the chunk was \emph{bad}.

\item \textbf{The parties store their hash values:} For all $i \in [2, \dots, n]$, $P_i$ stores $H_{k_i}(\psi_{i,0})$ in $\vec{\rho}_i$[0] and $P_1$, who is acting as $P$*, stores $\{\{H_{k_i}(\psi_{i,0}$*)$\}, H_{k_1}(\psi_{1,0})\}$ in $\vec{\rho}_1$[0], where $\{H_{k_i}(\psi_{i,0}$*)\} = \{$H_{k_1}(\psi_{1,0}$*), $H_{k_2}(\psi_{2,0}$*),..., $H_{k_n}(\psi_{n,0}$*)\}.

\item \label{first_direction} \textbf{$P$* directs the parties to move forward or rewind:} If the chunk was good, $P$* sends $Enc_1(F)$ to each $P_i$, and $P_1$ increments $\gamma_1$. Otherwise, $P$* sends $Enc_1(B1)$ to each $P_i$.

\item \label{updateChunk} \textbf{Each party updates its chunk and pebble:} For all $i \in [2, \dots, n]$, let $X_i$ be the symbol $P_i$ received from $P$* in Step \ref{first_direction}, depending on whether the chunk was good or bad. Now, each $P_i$ sets $\gamma_i \leftarrow \textsf{chunkUpdate}(X_i, \gamma_i, \alpha_i)$.

\end{enumerate}

\fbox{
 \begin{minipage}[h]{\textwidth}\textsf{chunkUpdate($X_i, \gamma_i, \alpha_i$):}
\begin{itemize}
\item If $X_i = F$, then increment $\gamma_i$.
\item If $X_i = B1$, then set $\alpha_i$ to point to the first node in $\mathcal{T}_i$ that is a node in the $\gamma_i$th chunk.
\item If $X_i = B2$, then set $\alpha_i$ to point to the first node in $\mathcal{T}_i$ that is a node in the $(\gamma_i - 1)$th chunk in $\mathcal{T}_i$. Then decrement $\gamma_i$.
\end{itemize}
\end{minipage}}

\subsection{The $j$th iteration} \label{secondTransfer}

\subsubsection{Chunk Simulation}

Suppose it is $P_\ell's$ turn to act as $P$*. Let $s$ be the index of the chunk of $\mathcal{T}$ that $P_\ell$ will simulate.  First, each $P_i$ sends $Enc_1(\gamma_i)$ to $P$*, which it decodes using $Dec_1$. Either $\gamma_i = s$ for all $i$, or there is some $i$ for which the two values differ, and $P$* acts accordingly:

\begin{itemize}
\item If $\gamma_i = s$ for all $i$, then $P$* requests $\{H_{k_i}(\alpha_{i,s-1}$*)$\}$ from the party it believes acted as $P$* during the simulation of chunk number $s-1$, say $P_k$ during time interval $j$. $P$* will request this set of hashes by sending $Enc_1(j)$ to $P_k$. If $P_k$ has already sent this set of hashes during a different chunk simulation, \emph{then he will ignore this message from $P$*.} This is crucial to prove the third part of Theorem \ref{main_thm}. Otherwise, he will send $Enc_2(\{H_{k_i}(\alpha_{i,s-1}$*)$\}$) to $P$*.

Meanwhile, $P$* will wait a specified amount of time for $P_k$ to respond to his message. If $P$* does not receive the set of hashes, then the transaction has ``timed out,'' and $P$* sends garbage symbols for the entirety of the chunk simulation. Otherwise, $P$* receives the set of hashes from $P_k$, and $P$* leads the chunk simulation following the protocol described in \cite{JKL}.

\item If $\gamma_i \not = s$ for some $i$, then there is no way for the chunk simulation to be successful. Therefore, $P$* sends garbage symbols for the entirety of the chunk simulation.
\end{itemize}

\subsubsection{Consistency Check}

After the chunk simulation, the following set of computations and exchanges will occur.

\begin{enumerate}

\item \textbf{Each party $P_i$ computes its hash value:} Let $\sigma_{i, \gamma_i}$ be the local transcript of the two-party protocol in $\mathcal{T}_i$ corresponding to the $\gamma_i$th chunk of $\mathcal{T}$. Set \[\psi_{i,\gamma_i} = \sigma_{i, \gamma_i}  \ || \ \vec{\rho}_i[\gamma_i - 1].\] Then $P_i$ then samples $k_i \leftarrow H$ and sends $Enc_1(H_{k_i}(\psi_{i,\gamma_i})  \ || \ k_i)$ to $P$* and sets $\vec{\rho}_i[\gamma_i] = H_{k_i}(\psi_{i,\gamma_i}).$ \label{piHash}

\item \textbf{$P$* computes its corresponding hash value:} For each $i \in \{1, \dots, n\}$, $P$* creates its version of the concatenated transcripts it received from $P_i$. Recall that at the beginning of the iteration, $P$* may have received the set $\{H_{\widetilde{k_i}}(\psi_{i,s-1}$*)$\}$, where from the party that $P$* believes most recently led the computation of chunk number $s-1$ of $\mathcal{T}$. (Here we use $\widetilde{k_i}$ to denote the key sampled by each $P_i$ during the previous iteration in question.) If $P$* didn't receive this set of hashes, either because it did not ask or because the operation timed out, then it skips to Step \ref{send_direction} and stores garbage values in $\vec{\rho}_i[s]$. Otherwise, let $\sigma_{i,s}$* be $P$*'s local transcript from the two-party protocol in $T_i$* corresponding to chunk number $s$ of $\mathcal{T}$. Set

\[\psi_{i,s}\text{*} = \sigma_{i,s}\text{*}   \ || \ H_{\widetilde{k_i}}(\psi_{i,s-1}\text{*}).\] Then $P$* applies the hash function $H_{k_i}$ to $\psi_{i,s}$* to get $H_{k_i}(\psi_{i,s}$*). Since $P_\ell$ is currently acting as $P$*, it stores $\{\{H_{k_i}(\psi_{i,s}$*)$\}_{i = 1}^n, H_{k_\ell}(\psi_{\ell,s})\}$ in $\vec{\rho}_\ell[s].$

\item \textbf{$P$* compares the hash values:} If $\widetilde{H}_{k_i}(\psi_{i,s}) = H_{k_i}(\psi_{i,s}$*$)$ for all $i \in \{1, \dots n\}$, then this chunk was \emph{good.} Otherwise, it was \emph{bad.}

\item \label{send_direction} \textbf{$P$* directs the parties to move forward or rewind:} If the chunk was good, then $P$* sends $Enc_1(F)$ to each $P_i$, and $P_\ell$ increments $\gamma_\ell$.

Otherwise, the chunk was bad. Recall that at the beginning of the chunk simulation, each $P_i$ sent $Enc_1(\gamma_i)$ to $P$*. Let $c_\beta$ be the smallest value of $\gamma_i$ for all $i \in \{1, \dots, n\}$, as computed by $P$*. If $\gamma_i = c_\beta$ and $\widetilde{H}_{k_i}(\psi_{i,s}) = H_{k_i}(\psi_{i,s}$*$)$, then $P$* sends $Enc_1(B1)$ to $P_i$. Otherwise, $P$* sends $Enc_1(B2)$ to $P_i$.  Finally $P_\ell = P$* sets $\gamma_\ell \leftarrow \textsf{chunkUpdate}(B1, \gamma_\ell, \alpha_\ell)$ if $c_\beta = \gamma_\ell$ and $\gamma_\ell \leftarrow \textsf{chunkUpdate}(B2, \gamma_\ell, \alpha_\ell)$ otherwise.

\item \label{secondUpdate} \textbf{Each party updates its chunk and pebble:} For all $i \in \{1, \dots, n\} \setminus \{\ell\}$, let $X_i$ be the symbol $P_i$ received from $P$* in Step \ref{send_direction}. Now each $P_i$ sets $\gamma_i \leftarrow \textsf{chunkUpdate}(X_i, \gamma_i, \alpha_i)$.

\end{enumerate}

\section{Measuring Progress}

In the following section, we prove the success of our simulation conditioned on the event that there are no hash collisions over the course of the computation. Since there are $\frac{5nL}{k}$ hash values computed over the course of the simulation, by Theorem \ref{hash_thm} and a union bound, the probability that there is no hash collision at any point during the computation is at least \[1 - \frac{5nL}{k}\cdot 2^{-N}.\] Therefore, if we want to ensure the success of the simulation with probability $1 - 2^{-\gamma}$ for some $\gamma$, we can set \[N \geq \gamma + \log(5nL).\] $t = qN$ is set accordingly.

Now, we must define a measure of progress that we can compute at each iteration during the protocol. We will show that by the end of the simulation, the measure of progress is sufficiently high as to ensure success. To this end, let $\xi$ be the node in $\mathcal{T}$ where the first error occurred. Say $\xi$ is in chunk $c_\xi$. Now, let $c_\beta$ be the chunk such that:
\begin{itemize}
\item $c_\beta \leq c_\xi$,
\item the party who most recently led the correct simulation $c_\beta$ has not yet been tapped for the set of hashes from $c_\beta$, and
\item for all chunks $c_k \leq c_\beta$, every party is in agreement about which party simulated $c_k$ most recently and during which time interval the simulation occurred.
\end{itemize}
Intuitively, if there have been errors, $c_\beta$ is the chunk that all of the parties must back up to in order to continue the correct simulation of $\pi$. Recall that we equip each party $P_j$ with a variable $\gamma_j$ that allows them to keep track of which chunk of $\mathcal{T}$ he thinks the protocol execution is currently in.

We are concerned with two types of error:

\begin{enumerate}
\item \textbf{External chunk error:} For some $j \in [1, \dots, n],$ $\gamma_j \not= c_\beta$.

\item \textbf{Internal chunk error:} For all $j \in [1, \dots, n]$, $\gamma _j = c_\beta$, but errors have occurred in the execution of chunk $c_\beta$.
\end{enumerate}

Suppose that there are external chunk errors. Let $P_j$ be the party ``furthest ahead" of the other parties, so $\gamma_j \geq\gamma_k$ for all $k \not=j$. Then there must be at least $\gamma_j - c_\beta$ consistency checks before all parties have backed up to chunk $c_\beta$, and can start making progress again. Therefore, we define the measure of progress to be \[M = c_\beta - (\gamma_j - c_\beta) = 2c_\beta - \gamma_j.\]

First, we give a lower bound on the number of errors injected in an iteration that could cause $M$ to decrease. Let $\epsilon'$ be the constant from Theorem \ref{JKL_thm} such that the protocol in \cite{JKL} is resilient to an $\frac{\epsilon'}{n}$-fraction of error. Recall that the smallest codeword sent outside of the chunk simulation protocol from \cite{JKL} in any iteration of $\widetilde{\pi}$ is $2\beta t$ bits, and that our encoding and decoding algorithms have constant relative distance $\lambda$. Then it takes $2\lambda\beta t$ bit flips to corrupt one word, and since there are at most $10n \beta t$ bits exchanged outside of the \cite{JKL} protocol in each iteration, these exchanges are resilient to a $\frac{\lambda}{5n}$-fraction of errors. Therefore, if we set $\epsilon = \min \{\epsilon', \lambda/5\}$, then both the chunk simulation protocol and the other exchanges during the iteration are resilient to an $\frac{\epsilon}{n}$-fraction of errors.

To prove the correctness of the simulation, we first analyze the change in $M$ during ``good'' and ``bad'' iterations. We say that an iteration is \emph{good} if the subprotocol from \cite{JKL} is not overwhelmed by errors and if every party correctly decodes every other message sent during the iteration. Otherwise, we say the iteration is \emph{bad}.

Throughout the following analysis, let $P_j$ be the party with a maximum $\gamma_j$ value.

\begin{claim} \label{claim_good}
A good iteration increases $M$ by at least 1.
\end{claim}
\begin{proof}
If $\gamma_j = c_\beta$, then a good iteration will increment $c_\beta$, so $M$ will increase by 1. If $\gamma_j \not= c_\beta$, then $P_j$ will decrement $\gamma_j$ by 1. In this case, $c_\beta$ will remain the same. To see this, note that if there is any discrepancy among the parties regarding their $\gamma_i$ values, then $P$* will not request the set of hashes from any party, so $c_\beta$ will not shift. Moreover, if $\gamma_i = \gamma_k$ for all $i,k \in [1, \dots, n]$, then $P$* may request a hash set from a previous $P$*, but since $\gamma_i \not= c_\beta$, that exchange will not effect the value of $c_\beta$. Therefore, $M$ will increase by 1.
\end{proof}

%
%

\begin{claim} If the iteration is \emph{bad}, $M$ will decrease by at most 3. \end{claim}

\begin{proof}

$c_\beta$ may move back by at most one during any iteration. This may happen in the case that either

\begin{itemize}

\item Some $P_i$ such that $\gamma_{i} = c_\beta$ jumps back when he should hold. In this case $c_\beta$ decreases by 1.
\item The party who had most recently led the correct simulation of $c_\beta$ was tapped for his hashes, and therefore will never send them again. If the simulation of $c_\beta$ is not successful, then $c_\beta$ will decrease by one.
\end{itemize}

Also, $\gamma_j$ can increase by at most one during any iteration, in the case where $P_j$, or some $P_k$ such that $\gamma_{k} = \gamma_{j}$, moves forward when he should jump back.

Therefore, $M$ decreases by at most 3.
\end{proof}

We are now ready to prove the resiliency of our simulation.

\begin{lemma} \label{measure_lemma}
When running $\widetilde{\pi}$ over a channel that makes at most \[E = \frac{5 \epsilon L}{8n}\] adversarial errors, $\widetilde{\pi}$ correctly simulates $\pi$.
\end{lemma}

\begin{proof}
The simulation is stopped after a total of $5L$ symbols have been exchanged, where $CC(\pi) = L$. Recall that $k$ is the depth of a chunk and $m$ is a constant positive integer such that it takes $mk$ exchanges to simulate a chunk via the protocol in \cite{JKL} and at least $4n \beta t = mk$ exchanges to complete the remainder of an interaction in $\widetilde{\pi}$. We will show that \[c_\beta > \frac{L}{k}\] by the end of the simulation, ensuring its correctness. To this end, we define the potential function \[\varphi = (2c_\beta - \gamma_{j})k + \frac{4n}{\epsilon m}E.\] (We abuse notation above and use $E$ to refer to an accumulating variable for the number of errors so far in the computation.)

Let $\varphi_\ell$ denote the change in $\varphi$ in iteration $\ell$, $\ell \in \{1, \dots, \frac{5L}{k}\}.$ We claim that $\varphi_\ell \geq k$ for all $\ell$. From Claim \ref{claim_good}, we know that if iteration $\ell$ is good, then $M= 2c_\beta - \gamma_j$ increases by one, so $\varphi_\ell \geq k$. Meanwhile, if iteration $\ell$ is bad, then $M$ decreases by at most three. However, this means that there were at least $\frac{\epsilon mk}{n}$ errors, enough to overwhelm the chunk simulation protocol from \cite{JKL}, the other exchanges that occurred during the $\ell$th iteration, or both. Therefore, $E$ increased by at least $\frac{\epsilon mk}{n}$, so \[\varphi_\ell \geq -3k + 4k = k.\]

Since there are $\frac{5L}{k}$ iterations during the simulation, $\varphi \geq 5L$ by the time it finishes. Now, we know that at the end of the simulation, $(2c_\beta - \gamma_j)k \geq 5L/2$ or $\frac{4n}{\epsilon m}E \geq 5L/2$. However, the latter would imply that \[E \geq \frac{5\epsilon L m}{8n} > \frac{5\epsilon L}{8n},\] which is a contradiction. Therefore, \[\frac{5L}{2} \leq (2c_\beta -
 \gamma_j)k \leq 2c_\beta k,\] so \[\frac{L}{k} < \frac{5L}{4k} \leq c_{\beta},\] as desired.

\end{proof}

We conclude with a proof of Theorem \ref{main_thm}.

\begin{proof}[Proof (Theorem \ref{main_thm})]
The first part of Theorem \ref{main_thm} is clear by construction, since we require that $CC(\widetilde{\pi}) = 5L = 5 CC(\pi)$. The second part of Theorem $\ref{main_thm}$ follows from Lemma \ref{measure_lemma}. Finally, each party will be required to serve as $P$* for no more than $\lceil \frac{5L}{kn} \rceil$ chunks and also must forward his hash sets to another $P$* no more than $\lceil \frac{5L}{kn} \rceil$ times. Therefore, the increase in communication required when simulating $\widetilde{\pi}$ is split evenly among the parties, so the communication complexity of any party $P_i$ is $\Theta\left(\frac{CC(\pi)}{n}\right)$.
\end{proof}
\bibliographystyle{alpha}
\bibliography{sources}
\end{document}